\documentclass[12pt]{article}
\usepackage[mathscr]{eucal}
\usepackage{amsfonts}
\usepackage{amssymb}
\usepackage{amsthm}
\theoremstyle{plain}
\newtheorem{thm}{Theorem}
\usepackage{epsf}
\newcommand{\be}{\begin{eqnarray}}
\newcommand{\ee}{\end{eqnarray}}
\newcommand{\ba}{\begin{array}}
\newcommand{\ea}{\end{array}}
\newcommand{\p}[1]{(\ref{#1})}
\newcommand{\lb}[1]{\label{#1}}

\def\bbox{{\,\lower0.9pt\vbox{\hrule \hbox{\vrule height 0.2 cm
\hskip 0.2 cm \vrule height 0.2 cm}\hrule}\,}}
\newcommand{\dsl}{\pa \kern-0.5em /}

\newcommand{\nn}{\nonumber \\}

\newcommand{\vecg}[1]{\mbox{\boldmath $#1$}}
\newcommand{\vecb}[1]{{\bf #1}}

\begin{document}


\begin{titlepage}

\vfill
\vfill

\begin{center}
\baselineskip=16pt {\Large  Taming the Zoo of Supersymmetric Quantum Mechanical  Models}\vskip 0.3cm {\large {\sl }}
\vskip 10.mm {\bf 
A.~V. Smilga}
 \\
\vskip 1cm

SUBATECH, Universit\'e de Nantes, \\
4 rue Alfred Kastler, BP 20722, Nantes 44307, France
\footnote{On leave of absence from ITEP, Moscow, Russia}\\
E-mail:  smilga@subatech.in2p3.fr

\end{center}
\vskip 2cm

\par
\begin{center}
{\bf ABSTRACT}
\end{center}
\begin{quote}

We show that in many cases nontrivial and complicated supersymmetric quantum mechanical (SQM) 
 models can be obtained from the simple model describing
free dynamics in flat complex space by two operations: {\it (i)} Hamiltonian reduction 
and {\it (ii)} similarity transformation
of the complex supercharges. We conjecture that it is true for {\it any} SQM model.

\end{quote}
\end{titlepage}
\setcounter{equation}{0}
\section{Introduction}

It is a common belief now that whatever the Grand Unified Theory is, it is a version of supersymmetric
field theory. Incorporating supersymmetry is the only known natural way to resolve the hierarchy problem.

An excellent playground to study supersymmetric dynamics is provided by SQM models. The simplest nontrivial such model 
introduced in \cite{Wit81} has the supercharges
 \be
\lb{QWit}
Q = \psi [p + i W'(x) ] \, , \ \ \ \ \ \ \ \ \ \ \  \bar Q = \bar\psi [p - i W'(x) ]
 \ee
and the Hamiltonian
 \be
\lb{HWit}
 H = \frac  12 \left[ {p^2} +  (W')^2 +   W''(x) (\bar \psi  \psi - \psi\bar\psi)  \right] \, .
 \ee
 In classical theory, $p,x$ are the usual conjugated phase space variables and $\bar \psi, \psi$ are canonically conjugated 
Grassmann variables.  In quantum theory (which we will be mainly concerned with), $x$ and $\psi$ are still usual real 
and complex Grassmann numbers,
while  $p$ and $\bar \psi$   become differential operators, $p = -i\partial/\partial x$,
$\bar\psi = \partial/\partial \psi$.

 $W(x)$ is an arbitrary function. The operators \p{QWit}, \p{HWit} satisfy the  simplest supersymmetry algebra,
  \be
\lb{algsusy}
Q^2 = \bar Q^2 = 0\, , \ \ \ \ \ \ \ \{\bar Q, Q \} = 2H\, .
  \ee

Since \cite{Wit81}, a lot of other models have been constructed. Many such models have a rather complicated form.
Some of them have extended supersymmetries --- several pairs of complex supercharges $Q_a, \bar Q_a$ satisfying the algebra
\footnote{Following the commonly adopted nowadays convention, ${\cal N}$ denotes the total number of {\it real} conserved 
supercharges. For the models with physical supersymmetry of the spectrum that involve (at least) double degeneracy of 
all excited states, 
${\cal N}$ is always even.}
 \be
\lb{algN}
\{Q_a, Q_b \} =  0\, , \ \ \{Q_a, \bar Q_b \} = 2 \delta_{ab} H \, , \ \ \ \ \ \ \ \ \ \ \ \ a,b = 1,\ldots, \frac 
{\cal N} 2 \, .
  \ee
The models involving up to 8 such pairs are known. 

One can write the action as an integral over the usual $(t, \theta, \bar\theta)$ or extended $(t, \theta_j, \bar\theta_j)$ 
superspace of a usual or extended real superfield expressed via fundamental superfields. Such action is manifestly 
invariant under supersymmetry transformations.
The problem is, however, that {\it many} such superfields and {\it many} such invariant actions can be constructed.

In this paper, we suggest an alternative approach. Instead of working in superspace, we restrict ourselves with supercharges
and Hamiltonians expressed in components. Then we observe that, in the all studied cases, an SQM system can be obtained 
from the basic
simple system describing the free flat complex dynamics,
 \be
\lb{freed}
  Q = \sqrt{2} \psi_a \pi_a, \ \  \bar Q = \sqrt{2} \bar \psi_a \bar \pi_a, \ \ \  \ \ \ \ \ H = \bar \pi_a \pi_a \ , \nn
(a = 1,\ldots,d )\, ,
 \ee
where $d$ is the complex dimension.

It is achieved by a combination of two operations: {\it (i)} Hamiltonian reduction  and {\it (ii)} 
similarity transformation of supercharges.

As a warm-up, let us obtain in this way the model \p{QWit}, \p{HWit}. We start from the model \p{freed} with $d=1$. 
The complex momentum $\pi$ has the real and imaginary parts, $\pi = (p_x + i p_y)/\sqrt{2}$. 
The wave functions depend on $x,y$ 
and the Grassmann holomorphic variable $\psi$.
At the first step, we impose the constraint $p_y \Psi = -i \partial \Psi/\partial y = 0$. We are allowed to do it as 
$p_y$ commutes with the Hamiltonian. The reduced Hamiltonian is just $p_x^2/2$. 

The constraint commutes not only with the Hamiltonian, but also with the supercharges. This implies that the 
reduced system enjoys
the same ${\cal N} = 2$ supersymmetry as the parent one. The reduced supercharges are 
$$Q^{\rm free} = p_x \psi, \ \bar Q^{\rm free}  = p_x \bar \psi \, .$$ 
It is just the free ($W=0$) version of Witten's model \p{QWit}, \p{HWit}.

The potential can be introduced at the second step by a similarity transformation. Indeed, the supercharges \p{QWit}
can be expressed via the free ones as 
 \be
\lb{simil}
Q = e^W Q^{\rm free} e^{-W} \, , \ \ \ \ \ \ \ \ \ \bar Q = e^{-W} \bar Q^{\rm free} e^{W} \, .
  \ee
The transformed supercharges $Q, \bar Q$ are nilpotent if $ Q^{\rm free},  \bar Q^{\rm free}$ are nilpotent 
and they hence satisfy  the same
supersymmetry algebra. Note that this is a similarity transformation for the {\it supercharges}. 
The operator $e^W$ is not unitary such that  $Q$ and $\bar Q$ are transformed in a different way. 
As for the Hamiltonian $\{\bar Q, Q\}$, it is not related to the free Hamiltonian
by any similarity transformation and has a distinct spectrum.

This example is trivial, but we will see in the next two sections that this philosophy works in many other not 
so trivial cases and a complicated
SQM model can be obtained from the free model \p{freed} 
by performing a proper Hamiltonian reduction and a proper similarity transformation.

\section{${\cal N} = 2$ sigma models.}

\subsection{Dolbeault complex}
Let us concentrate on $Q$ and perform the following similarity transformation of the free supercharge in \p{freed},
 \be
\lb{similDol}
Q = e^R Q^{\rm free} e^{-R} 
  \ee
where $R$ is not just a function of coordinates as in \p{simil}, but an operator, $R = \omega_{ab} \psi_a \bar \psi_b$.
The supercharge $\bar Q$ will then be rotated with the operator $e^{-R^\dagger}$. 
(It will be convenient for us later to introduce $\bar Q$ rotated with an extra scalar function reflecting 
the presence of a nontrivial Hilbert space measure in the rotated system --- see Eq.\p{QbarDolb} below. 
But let us keep for a moment
$Q$ and $\bar Q$ Hermitially conjugate in the naive sense, without taking into account the measure.)    
When $\omega_{ab}$ is anti-Hermitian, $e^R$ is unitary, $Q,\bar Q$, and $H$ are rotated by the same operator, and 
this boils down to a canonical transformation of the phase space variables. On the other hand, when $\omega_{ab}$
is Hermitian, the supercharges $Q$ and $\bar Q$ are transformed differently, and their anticommutator is nontrivial.

The calculation can be done using the Hadamard formula,
 \be
\lb{Hadamard}
e^R X e^{-R} \ =\ X + [R,X] + \frac 12 [R, [R,X]] + \ldots
 \ee
In our case, this implies
 \be
\lb{Hadappl}
e^R \psi_c e^{-R} &=& \psi_a \left( e^\omega \right)_{ac} \, , \\
e^R \partial_c e^{-R} &=& \partial_c + \left( e^\omega \right)_{ae} \left( \partial_c  e^{-\omega} \right)_{eb} 
\psi_a \bar \psi_b \, .
 \ee
We thus derive  
 \be
\lb{Qomega}
Q \ = \ \sqrt{2} \psi_d \left( e^\omega \right)_{dc} \left[ \pi_c - i \left( e^\omega \right)_{ae} 
\left( \partial_c  e^{-\omega} \right)_{eb} \psi_a \bar \psi_b \right] \, . 
 \ee
The associated Hamiltonian has the kinetic term with a nontrivial Hermitian metric,
 \be
\lb{Hkin}
H^{\rm kin} = \left(e^{\omega^\dagger} e^{\omega} \right)_{ab} \bar \pi_a \pi_b  \to  
\left(e^{\omega^\dagger} e^{\omega} \right)^{\bar k j} \bar \pi_{\bar k} \pi_j \, .
  \ee 
The matrices $e^{\pm \omega}, e^{\pm \omega^\dagger}$ can then be interpreted as the complex {\it vielbeins},
 \be
\lb{vielbomega} 
 \left( e^\omega \right)_{ac} \to e^j_a, \ \ \ \ \ \ \left( e^{-\omega} \right)_{ca} \to  e_j^a , \ \ \ \ \ \ \ 
  \left( e^{\omega\dagger} \right)_{ca} \to e^{\bar j}_{\bar a}, \ \ \ \ \ \ \ 
 \left( e^{-\omega\dagger} \right)_{ac} \to e_{\bar j}^{\bar a}
 \ee
When $\omega$ is Hermitian, the vielbein matrix $e^j_a$ is also Hermitian. For generic $\omega$, the vielbein is a 
generic complex matrix,
with the anti-Hermitian part of $\omega$ corresponding to tangent space rotations.

One can note now that the supercharge \p{Qomega} can be rewritten as
\be
\lb{QDolb}
Q \ =\  \sqrt{2} \psi^j \left( \pi_j  + i\Omega_{j, \bar b a} \psi_a \bar \psi _{ b} \right)
 \ee
with $\Omega_{k, \bar b a}$ being the so called  {\it Bismut spin connection} corresponding to the metric 
$h = e^{\omega\dagger} e^{\omega}$ and the vielbein $e = e^\omega$.
\footnote{The Bismut spin connection  is related 
to the Bismut affine connection\cite{Bismut}, which is a torsionfull affine connection such that {\it (i)} 
the covariant derivatives of the metric and the complex structure
matrix vanish; {\it (ii)} the torsions are completely antisymmetric. 
See \cite{ISDir,RR} for all definitions  and  notations. To avoid a confusion, note also 
here that the Bismut spin connection does {\it not} coincide with the structure $-e_a^l (\partial_j e_l^b)$ entering
\p{Qomega}, but involves extra terms. These terms vanish when multiplying by $\psi^j \psi_a$. 
Cf. Eq.(3.13) of Ref.\cite{ISDir}.} 

A nontrivial metric introduces a natural covariant measure in the Hilbert space,
 \be
\lb{measure}
\mu \ =\ \det h \, \prod_{j} dz^j d\bar z^j \ .
 \ee
It is convenient to define $\bar Q$ to be Hermitially conjugate to $Q$ with respect to this measure,
 \be
\lb{QbarDolb}
\bar Q =   (\det h)^{-1} Q^\dagger \, \det h  \, ,
 \ee
where $Q^\dagger$ is the ``naive'' Hermitian conjugation.

The supercharges thus obtained exactly {\it coincide} with the supercharges (3.26) in Ref.\cite{ISDir}, if setting there 
$W = \ln \det h/4$. These supercharges were obtained from N\"other supercharges of a certain SQM model 
with a nontrivial superspace 
Lagrangian \cite{Hull,ISDir},
 \be
\lb{LDolb}
 L = -\frac 14 h_{j\bar k} D Z^j \bar D \bar Z^{\bar k} + W(Z,\bar Z)\, ,
  \ee
where $Z^j$ and $\bar Z^j$ are chiral $d=1$ superfields and $D, \bar D$ are supersymmetric covariant derivatives. 

As was explained in details  in \cite{ISDir}, the Hilbert space of the functions $\Psi(z^j, \bar z^j, \psi_a)$ with the measure
\p{measure} can be mapped onto the space of holomorphic $p$-forms realizing the Dolbeault complex.
The supercharges \p{QDolb}, \p{QbarDolb} can then be mapped to the exterior holomorphic derivative operator 
$\partial$ and its Hermitian
conjugate.

As was just mentioned, the supercharges \p{QDolb} and \p{QbarDolb} correspond to a particular choice of $W$ in the Lagrangian 
\p{LDolb}. One can, however, obtain the model with any $W$ 
by applying an extra similarity transformation,
  \be
\lb{similW}
Q \to e^G Q e^{-G} 
 \ee
with $G = W - \frac 14 \ln \det h$. 
\footnote{For a compact complex manifold, there is no global expression for the  metric valid everywhere. 
One should introduce {\it charts}. 
The expressions like $\ln \det h$ are in fact not nonsingular {\it functions} on the manifold, but should be understood 
as {\it sections} 
of a certain line (Abelian) fiber bundle. Matching the local expressions for these sections in the regions where the 
charts overlap imposes the restrictions on $W$ associated with the quantization of topological charge. 
If these restrictions are not fulfilled, supersymmetry
is lost \cite{flux}.}  This gives what mathematicians call a {\it twisted  holomorphic Dolbeault complex}

Another  distinguished choice, besides $W =  \frac 14 \ln \det h$ corresponding to
$G =  0$, is $W =  - \frac 14 \ln \det h$  corresponding to
$G = - \frac 12  \ln \det h$ . This amounts to the overall similarity transformation with
$\exp\{ -\omega_{ab} \bar \psi_b \psi_a\}$ and gives 
the  {\it antiholomorphic}  untwisted Dolbeault complex.

In the physical language, this complex describes the dynamics of a Dirac operator in the presence of Abelian 
gauge field, $A_M = \{-i \partial_m W, i\partial_{\bar m} W \}$. The models with non-Abelian gauge fields 
can  be obtained from the model \p{LDolb} with $W=0$ by a similarity transformation \p{similW} with a 
matrix-valued $G$. 

The supercharge $Q$ can be further rotated with a {\it holomorphic} 
\footnote{It is holomorphic with respect to fermion variables, but ${\cal B}_{jk}$ etc. are arbitrary functions
of $z^j$ and $\bar z^j$. }
operator 
  \be
\lb{Rholomorph}
\exp \left\{{\cal B}_{jk} \psi^j \psi^k + {\cal B}_{jklm} \psi^j \psi^k \psi^l \psi^m + \ldots \right \}
 \ee
One obtains in this way complex sigma models with torsions studied in \cite{FIS}.

\subsection{De Rham complex}
Note first that, representing $\pi_a = [p_x^{(a)} + ip_y^{(a)}]/\sqrt{2}$ in \p{freed} and imposing $d$ constraints 
$p_y^{(a)} \Psi = 0$, we obtain the model describing free flat {\it real} dynamics. 
The supercharges are
 \be
\lb{Qreal}
 Q = p_A \psi_A, \ \ \ \ \ \ \ \ \ \ \ \ \, \ \bar Q =  p_A \bar\psi_A \, .
  \ee 
(with $p_A \equiv p_x^{(a)}$, $A = 1,\ldots,d$).

Let us apply now a similarity transformation \p{similDol} with
\be
\lb{RAB}
 R = \omega_{AB} \psi_A \bar \psi_B \, .
 \ee
When $\omega_{AB}$ is anti-Hermitian, this amounts to a unitary rotation. New nontrivial models are obtained
for Hermitian  $\omega_{AB}$.

Let first  $\omega_{AB}$ be {\it real} and symmetric. By the same token as in the complex case, one obtains
\be
\lb{QRham}
Q \ = \  \psi_D \left( e^\omega \right)_{DC} \left[ p_C - i \left( e^\omega \partial_C  e^{-\omega}  \right)_{AB} 
 \psi_A \bar \psi_B \right] \equiv \  
 \psi^M \left( p_M  - i\Omega_{M, AB} \psi_A \bar \psi_B \right) \, ,
 \ee
where $\psi^M = e^M_A \psi_A$ with the real vielbeins 
 \be
\lb{vielreal}
 e^M_A = \ \left( e^\omega \right)_A^{\ M},\ \ \ \ \ \ \ \ \ \  e_{MA} = \ \left( e^{-\omega} \right)_{MA}
  \ee
giving the metric 
 \be
 \lb{metr}
g_{MN} =  \left( e^{-2\omega}\right)_{MN} \, .
  \ee
Profiting from anticommutativity of $\psi^M$ and $\psi_A$, we have expressed the 3-fermion structure in the supercharge
via  the spin connection 
 \be
\lb{Om}
 \Omega_{M, AB} = e_{AN} \left( \partial_M e^N_B + \Gamma^N_{MK} e^K_B \right)\, . 
 \ee

The supercharge \p{QRham} is well known. It can be mapped to the  exterior 
derivative operator $d$ of the de Rham complex \cite{Witgeom}. The supercharge $\bar Q$ is convenient
to define as $\bar Q =\ ( \det g)^{-1/2} Q^\dagger  ( \det g)^{1/2}$ such that it is Hermitially conjugate to $Q$ with the
Riemann covariant measure $\sqrt{  \det g}$. It can be presented as
 \be
\lb{QbarRham}
\bar Q \ = \  
 \bar \psi^M \left( p_M  - i\Omega_{M, AB} \bar \psi_A  \psi_B \right) \, .
 \ee

The Lagrangian of this model can be easily 
written in terms of $d$ real superfields 
\be
\label{Xj}
X^M = x^M + \theta \psi^M  +
 \bar \psi^M  \bar \theta + F^M  \theta \bar \theta \, .
 \ee   
It has the form \cite{FT} 
\be
\label{121stand}
L = \frac 12 \int d\theta d \bar\theta  \, g_{MN}(X) DX^M \bar D X^N \,. 
 \ee

The de Rham complex can be deformed by adding the potential, which amounts to a similarity 
transformation $Q \to e^W Q e^{-W} $
(in contrast to the complex case, here $W$ must be a well-defined scalar function) or adding torsions \cite{MWu,FIS} 
 which amounts to a 
holomorphic
similarity transformation with the operator $\exp\{{\cal B}_{MN} \psi^M \psi^N + \cdots \}$.  

Consider now the case of generic Hermitian $\omega_{AB}$. The supercharges can, again, be represented as in
\p{QRham}, but neither the vielbeins \p{vielreal} nor the ``metric'' \p{metr} are real anymore.
As is shown in \cite{quasi}, the corresponding Lagrangian has the same form as in
\p{121stand}, but with the complex Hermitian $g_{MN}$. Complexity of the metric means that  
this system cannot thus be interpreted anymore as a de Rham complex. It is something new. 

It is further shown in \cite{quasi} that this new complex
 can be obtained by a Hamiltonian
reduction of the Dolbeault complex for some special complex manifolds whose metric does not
depend on imaginary parts of complex coordinates.
Indeed, consider a manifold of the complex dimension $d$ with isometries corresponding to the imaginary
coordinate shifts.  We can then impose $d$ constraints
$G_j \Psi = \partial \Psi /\partial ({\rm Im}\, z^j) = 0$. These constraints commute with the Dolbeault
Hamiltonian 
(allowing for the Hamiltonian reduction) and the supercharges (meaning that the reduced 
system enjoys the same supersymmetry  as the original one). A not so difficult analysis 
 shows that, after such reduction,
the supercharge \p{QDolb} (it is convenient to write it in terms of the fermion variables with world indices  $\psi^j$)
goes over into \p{QRham}, written in terms of $\psi^M, \bar\psi^M$.  

In other words, the model of this type (we called it
  a {\it quasicomplex} sigma model  \cite{quasi}) can be obtained
from the basic model \p{freed} by a subsequent application of two operations: similarity transformation and Hamiltonian
reduction. And, as illustrated in Fig.1, the result does not depend on the {\it order} 
in which these operations are performed.

\begin{figure}
\label{hren}
 \begin{center}
        \epsfxsize=200pt
        \epsfysize=0pt
        \vspace{-5mm}
        \parbox{\epsfxsize}{\epsffile{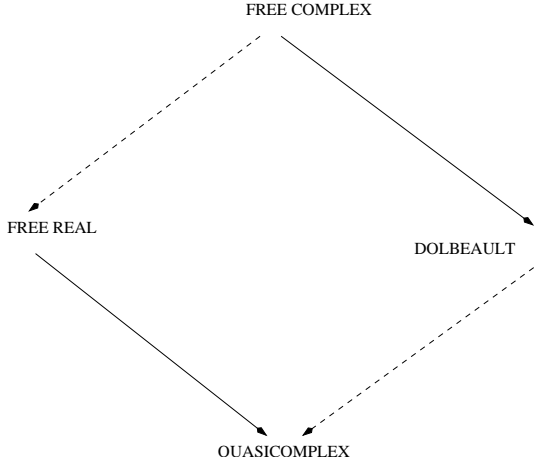}}
        \vspace{5mm}
    \end{center}
\caption{A rhombus of sigma models. 
The solid arrows stand for a similarity transformation and the dashed 
arrows --- for a Hamiltonian reduction. }
\end{figure}

We would also like to note here that the supercharges $Q = \sqrt{2} \psi_a \pi_a$ and $Q = \psi_A p_A$ 
can in principle be rotated by a similarity transformations with {\it antiholomorphic} operators like 
\be
\lb{antihol}
R \ =\ \exp\{ {\cal C}_{AB} \bar \psi_A \bar \psi_B \}\, .
 \ee
 Such models were never considered,  and it would be interesting to do so. 

\section{Extended supersymetries.}
An arbitrary similarity transformation \p{similDol} leaves the supercharge $Q$ nilpotent and hence keeps
supersymmetry. For a model with extended supersymmetries, it is not always the case. I.e. the minimal
${\cal N} = 2$ supersymmetry is always kept, but if we want to preserve extended supersymmetries, the operator 
of similarity transformation should satisfy certain extra conditions. The same concerns the Hamiltonian reduction
procedure. If we want the reduced model to keep all the supersymmetries of the original one, the constraints should
commute with all supercharges. 

Consider some examples.

\subsection{K\"ahler sigma models.}
Consider the real free dynamics with the supercharges \p{Qreal} and the Hamiltonian $H = p_A^2/2$. Assume that the
dimension $D$ is even. It is easy to see that one can add extra pairs of nilpotent supercharges whose anticommutator
gives the same Hamiltonian and which commute with $Q, \bar Q$. They  form thereby together
with \p{Qreal} an extended supersymmetry algebra. Each such pair of supercharges can be represented as
  \be
\lb{Sreal}
 S = p_A I_{AB} \psi_B, \ \ \ \ \ \ \ \ \ \ \ \ \, \ \bar S =  p_A  I_{AB} \bar\psi_B \, .
  \ee     
where $I_{AB}$ is a real antisymmetric matrix satisfying the condition $I^2 = -1$. 
\footnote{It is, of course, recognizable as a 
flat complex structure matrix.}  

Suppose that there is only one such extra pair. 
 It is convenient to introduce complex coordinates $\{z_{a = 1, \ldots, D/2}, \bar z_{a = 1, \ldots, D/2} \}$ 
(the eigenvectors of $I$), 
trade $\psi_A$ for $\chi_{a\alpha}$, $\alpha = 1,2$ 
and deal with the supercharges 
\footnote{When $D=2$ and $I_{AB} = \epsilon_{AB}$, the explicit conventional form of the combinations
entering \p{TTbar} is  $$\chi_{1} = \frac 1 {\sqrt{2}} (\psi_{1} + i \psi_{2}), \ \ \ 
 \chi_{2} = \frac 1 {\sqrt{2}} (\bar \psi_{1} + i \bar \psi_{2}), \ \ \ \ \ \pi  =  \frac 1 
{\sqrt{2}} (p_1 - ip_2) \, $$
and complex conjugates.  
This is trivially generalized to any even $D$, if choosing $I = {\rm diag} (\epsilon, \ldots, \epsilon)$.}
  \be
\lb{TTbar}
 T_\alpha \ =\ \sqrt{2} \pi_a \chi_{a\alpha}, \ \ \ \ \ \ \ \ \ \ \ \  
 \bar T_\alpha \ =\ \sqrt{2} \bar \pi_a \bar \chi_{a\alpha} \, .
  \ee
They represent the following linear combinations of $Q,S, \bar Q, \bar S $,
 \be
\lb{lincomb}
 T_1 \ =\ \ \frac {Q + i S}{\sqrt{2}}  , \ \ \ \ T_2 =  \frac {\bar Q + i\bar S}{\sqrt{2}} , 
\ \ \ \ \bar T_1 = 
 \frac{\bar Q - i\bar S}{\sqrt{2}}, \ \ \ \ 
\bar T_2 =  \frac {Q - iS}{\sqrt{2}} \, .
 \ee
Let us perform now 
 a similarity transformation 
\be
\lb{similKahl}
T_\alpha  \ \to \  e^{\omega_{ab} \chi_{a\beta} \bar \chi_{b\beta}} T_\alpha   
e^{-\omega_{ab} \chi_{a\beta} \bar \chi_{b\beta}} 
  \ee
with a Hermitian $\omega_{ab}$. In the full analogy with \p{Qomega}, \p{QDolb}, we obtain
\be
\lb{QomegaKahl}
T_\alpha \ = \ \sqrt{2} \chi_{d\alpha} \left( e^\omega \right)_{dc} \left[ \pi_c - i \left( e^\omega \right)_{ae} 
\left( \partial_c  e^{-\omega} \right)_{eb} \chi_{a\beta} \bar \chi_{b\beta} \right] \, , 
 \ee
which can be represented as 
  \be
\lb{QKahl}
T_\alpha \ =\  \sqrt{2} \chi^j_\alpha \left( \pi_j  + i\Omega_{j, \bar b a} \chi_{a\beta} 
\bar \chi _{b\beta} \right) \, .
 \ee
However, in a generic case, the supercharges \p{QKahl} and their conjugates {\it do} not form the 
extended supersymmetry algebra --- 
the anticommutator $\{T_1, \bar T_2\}$ does not vanish and $\{T_2, \bar T_2\}$ does not 
coincide with $\{T_1, \bar T_1\}$. But,
for some {\it special} $\omega$ when the  metric 
$e^{\omega^\dagger } e^{\omega}$ is K\"ahler, the ${\cal N} = 4$ superalgebra holds 
\cite{McFar,howto}. 
In this case, 
$\Omega_{j, \bar b a} \ = e^{\bar k}_{\bar b} \partial_j e^{\bar a}_{\bar k}$
 entering \p{QKahl} are the standard torsionless spin connections.

As is written in \p{lincomb}, the supercharge $T_1$ is expressed via $Q$ and $S$, while the supercharge $T_2$ 
is expressed via $\bar Q$ and $\bar S$.
This means that 
 the similarity transformation \p{similKahl} of  $T_\alpha$ corresponds to a rather complicated 
transformation (not a similarity one)
 of the ``original'' supercharges $Q,S$. 

Alternatively, one can rotate, as we have seen, the flat supercharge $Q$ with 
the operator \p{RAB} 
to obtain the de Rham supercharge \p{QRham}. If the metric thus obtained is K\"ahler, 
the same similarity transformation applied to $S$  gives us 
the second pair of supercharges, keeping the ${\cal N} = 4$ supersymmetry.
Indeed, the result of such rotation of $S^{\rm flat}$ is
\be
\lb{Srot}
S^{\rm rotated} \ =\ \psi^M I_M^{\ N} \left(p_N - i \Omega_{N,  AB} \psi_A \bar \psi_B \right)\ .
\ee
An accurate proof of the fact that, for K\"ahler manifolds, the operators
\p{Srot} and \p{QRham} together with their conjugates satisfy the same 
commutation relations of the ${\cal N} = 4$ superalgebra 
as the flat supercharges is presented in the Appendix.

Note that this similarity transformation of $Q$ and $S$ corresponds to a complicated transformation of
$T_\alpha$ and $\bar T_\alpha$.    

\subsection{Hyper-K\"ahler sigma models.}

In  flat space of real dimension 4 or multiple integer of 4, $D = 4m$, 
one can write three additional pairs of supercharges,
  \be
\lb{S123}
 S^{1,2,3} = p_A I^{1,2,3}_{AB} \, \psi_B, \ \ \ \ \ \ \ \ \ \ \ \ \, \ \bar S^{1,2,3} =  
p_A  I^{1,2,3}_{AB} \, \bar\psi_B \, .
  \ee  
associated with three complex structures $I^{1,2,3}$ satisfying the quaternionic algebra.
 \be
\lb{quatern}
 I^a I^b = -\delta^{ab} +  \epsilon^{abc} I^c \ .
 \ee
By an orthogonal transformation, one can bring them into a canonical form
 \be
\lb{Iacanon} 
(I^a)_{AB} =   {\rm diag} \{ -\eta^a_{\mu\nu}, \ldots, -\eta^a_{\mu\nu} \}\, ,
 \ee 
where $\eta^a_{\mu\nu} $ are 't Hooft symbols.

The supercharges $Q, \bar Q, S^a, \bar S^a$ form the ${\cal N} = 8$ supersymmetry algebra, which the 
flat Hamiltonian $H=p_A^2/2$ thus enjoy.
In a special case when the metric of the manifold corresponding to the de Rham supercharges
 \p{QRham} obtained after a similarity transformation of the 
flat supercharge $Q$ is hyper-K\"ahler, the Hamiltonian thus obtained 
also admits 3 extra pairs of conserved supercharges \cite{A-GF}. 
Three of these six
extra supercharges are obtained from the flat supercharges $S^a$  in \p{S123} by the {\it same}
similarity transformation as the one
 applied to the supercharge $Q$.
Their explicit form is 
\be
\lb{Sarot}
S^a_{\rm rotated} \ =\ \psi^M I_M^{a\ N} \left(p_N - i \Omega_{N,  AB} \psi_A \bar \psi_B \right)\ .
\ee
The explicit proof of the fact that, in the hyper-K\"ahler case,
 the supercharges \p{Sarot} and their conjugates form together with the supercharge \p{QRham} and its conjugate
the standard ${\cal N} = 8$ superalgebra is given in the Appendix.

Alternatively, in the full analogy with the K\"ahler case, 
the supercharges \p{S123} together with \p{Qreal} can be 
rearranged  by defining
 \be
T_\alpha \ =\ \left( \gamma_\mu p_\mu^k \chi^k \right)_\alpha 
 \ee
and their conjugates $\bar T_{\alpha}$, In the expression above,
 $\mu, \alpha  = 1,2,3,4,\  k = 1,\ldots , m$, $\gamma_\mu$ are Euclidean 4-dimensional 
$\gamma$ - matrices and $\chi^k_\alpha$ are Dirac 4-component spinors.

One can then rotate $T_\alpha$ with a matrix 
$e^R = \exp\{ \omega_{kq} \chi^k_\alpha \bar \chi^q_\alpha \}$ and, when the
metric $e^{\omega^\dagger} e^{\omega}$ thus obtained is hyper-K\"ahler,  arrive at the hyper-K\"ahler
supercharges in the form written in \cite{Fig}. 

This similarity transformation of $T_\alpha$ corresponds to a complicated transformation 
of $Q, \bar Q, S^a, \bar S^a$. On the other hand, the similarity transformation of $Q$ and $S^a$ discussed above
corresponds to a complicated transformation of
$T_\alpha,\ \bar T_\alpha$.

\subsection{HKT and OKT.}

K\"ahler and hyper-K\"ahler sigma models represent special cases of the generic de Rham sigma model that admit
extra supercharges. There are also special complex Dolbeault sigma models admitting extra supersymmetries.
In particular, in flat complex space of even complex dimension $d=2m$, 
one can add to the supercharges \p{freed}, the supercharges
   \be
  \lb{extraSDol}
  S = \sqrt{2} \epsilon_{ab} \psi_a^k  \bar \pi_b^k, \ \  \bar S = \sqrt{2}  \epsilon_{ab}  \bar  
\psi_a^k \pi_b^k \, .  
   \ee 
Performing a similarity transformation \p{similDol} with a special class of $R$ respecting ${\cal N} = 4$ supersymmetry, 
one obtains the so called HKT sigma models \cite{HP,GPS}. 
To see that, consider the simplest $d=2$ case. By introducing real and imaginary
parts $\psi_{A=1,2,3,4}$ of $\psi_{a=1,2}$, 
we can bring \p{freed}, \p{extraSDol} to  a more familiar form including 
 four pairwise anticommuting Hermitian supercharges \cite{Wipf,HKT},
 \be
\lb{QHKT}
{\cal Q} \ = \ \psi_A p_A \, , \ \ \ \ {\cal S}^a \ =\  -\eta^{a}_{AB}\psi_A p_B\, , 
 \ee 
where $a = 1,2,3$ and $\eta^{a}_{AB}$ are `t Hooft symbols \cite{Hooft}.

This looks similar to \p{S123}, but the variables $\psi_A$ are now {\it Hermitian}. They
 satisfy the Clifford algebra, $\{\psi_A, \psi_B\} = \delta_{AB}$ and can be mapped to 
gamma  matrices. For our approach, we need, however the {\it holomorphic} supercharges $Q,S$ in \p{freed}, \p{extraSDol}.
It is they who are going to be rotated with a similarity transformation \p{similDol}, the supercharges $\bar Q, \bar S$ 
being transformed with $e^{-R^\dagger}$. 
Let us choose $R = g(x)\psi_a \bar \psi_a$.
We derive
  \be
\lb{QSconf}
Q \ \to \ \sqrt{2} f\psi_a \left( \pi_a + \frac {i\partial_a f}{f} \psi_c \bar \psi_c \right) \, , \nn 
S \ \to \ \sqrt{2} f \epsilon_{ab} \psi_a \left( \bar \pi_b + \frac {i\bar\partial_b f}{f} \psi_c \bar \psi_c \right) \, .
 \ee
with $f = e^g$. The corresponding metric is conformally flat, $ds^2 = e^{-2g} dx_\mu^2$, the simplest HKT metric
\footnote{The formal definition is the following \cite{Verb}. 
The manifold is called HKT  if it admits three complex structures satisfying \p{quatern} which are covariantly 
constant with respect to one and the same affine connection. Generically (and, in particular, for conformally 
flat 4-dimensional manifolds), this connection (the Bismut connection) involves torsions. 
In some special cases, the torsions vanish, the Bismut connection boils
down to the usual Levy-Civita connection and the HKT (i.e. {\it hyper-K\"ahler with torsion} ) 
manifolds boil down hyper-K\"ahler manifolds.}.     
The supercharges \p{QSconf}
(derived first in \cite{Maxim}) together with their conjugates
  satisfy the ${\cal N} =4$ superalgebra.

Dolbeault models enjoying ${\cal N} = 8$ supersymmetry are known as OKT manifolds. Their real 
dimension is a  multiple integer of 8. Indeed,  flat 8-dimensional space admits 8 anticommuting Hermitian 
supercharges: the supercharge ${\cal Q} = p_A \psi_A$, where $A=1,\ldots,8$ and $\psi_a$ 
are now {\it real}, and the supercharges
${\cal S}^{a = 1, \ldots, 7} = (\Gamma^a)_{AB} p_A \psi_B $, where $\Gamma^a$ are 7-dimensional 
real antisymmetric gamma matrices. One of the convenient representations for the latter is
  \be
\lb{Gamma7} 
 \Gamma^{1,2,3} \ =\ \left(  \begin{array}{cc} - \bar \eta^a & 0 \\ 0  & \bar \eta^a \end{array} \right)\, , \ \ \  
 \Gamma^{4,5,6} \ =\ \left(  \begin{array}{cc} 0  &  \eta^a \\  \eta^a & 0  \end{array}  \right) \, , \ \ \ 
   \Gamma^7 \ =\ \left(  \begin{array}{cc}  0 & 1 \\ -1 & 0 \end{array}  \right) , 
 \ee
One can also relate the matrix elements in \p{Gamma7} to the structure constants of 
octonion algebra\cite{GPS} --- that is why the name
OKT (octonionic K\"ahler with torsion) was choosen. 
\footnote{In contrast to HKT, the matrices \p{Gamma7} cannot, of course, satisfy  non-associative octonion algebra.
Moreover, one cannot choose among the matrices \p{Gamma7} three matrices satisfying the quaternionic algebra \p{quatern}. 
This means that an OKT manifold need not to be an HKT manifold. Note also that one can deform the 
flat model with breaking ${\cal N} = 8$ supersymmetry but keeping the ${\cal N} = 4$ supersymmetry 
associated with the unity matrix and, say, the matrices $\Gamma^{1,2,3}$. One obtains in this way a 
class of ${\cal N} =4$ models, so called {\it Clifford} models 
that {\it are} not HKT
\cite{GPS,Hull,DI,Fedoruk}. }
The simplest example of a nontrivial OKT manifold is a conformally flat 8-dimensional manifold where the conformal factor
$f(x^M)$ represents a harmonic function, i.e. $f(x^M) = 1 + C|x^M|^{-7}$. 

The OKT models represent a particular case of  Dolbeault models with extra holomorphic torsions 
(see  \cite{Fedoruk} for a detailed discussion).
Thus, they can be obtained, as discussed above, from the flat models of the corresponding dimension by a 
similarity transformation of the
supercharge $Q$ in \p{freed}. It remains to be seen whether one can conveniently define in this case three
 other complex supercharges  obtained from the flat ones by the  same similarity transformation.

\subsection{Reduced models.}
Consider an HKT model on a 4-dimensional conformally flat manifold with the supercharges \p{QSconf}. Suppose that the 
conformal factor $f(x^M)$ does not depend on one of the variables, say, $x^4$. It is then straightforward to observe that 
the operator $\hat p_4$ commutes with the supercharges and the Hamiltonian and can thus be used to 
perform Hamiltonian reduction.
 As a result, we obtain a ${\cal N} = 4$ supersymmetric QM model describing dynamics on a conformally 
flat 3-dimensional manifold.
This model \footnote{A ({\bf 3}, {\bf 4}, {\bf 1}) - model in the notation of \cite{PT}, where the first numeral stands 
for the number of bosonic dynamic degrees of freedom, the second - the number of fermionic d.o.f. 
and the third - the number of auxiliary
fields in superfield description.} was first constructed in \cite{1987} and described in superfield language in
\cite{IS91,BP}. Taking a $4m$ dimensional HKT with the metric not depending on $m$ variables, one obtains a generalised
$3m$ - dimensional model considered in \cite{IS91}.

One can, of course, consider many other HKT models  (or, in the language of \cite{PT}, the models with several
{\it root} \cite{Gates} 
({\bf 4}, {\bf 4}, {\bf 0}) multiplets) living on  manifolds with various isometries. 
Factorizing over these isometries gives
a multitude of models. Hamiltonian reduction of the model in 
a flat or conformally flat 4-dimensional  space with respect to its $U(1)$ isometry 
was considered in \cite{BelNerYer}. One obtains in such a way the ${\cal N} = 4$ models with an extra magnetic monopole
\cite{RitCrom,1987}.
For another example, one can take a flat ${\cal N} =4$ model endowed with a self-dual instanton
field \footnote{It was mentioned  above that a non-Abelian gauge field can be brought 
about by a matrix-valued similarity transformation. The transformations of this kind that give a self-dual gauge field
respect the ${\cal N} = 4$ supersymmetry of the flat model \cite{Maxim}.},  
 \be
 {\cal A}_\mu  =  \frac {2 \eta^a_{\mu\nu} x_\nu t^a }{x^2 + \rho^2}
 \ee
The supercharges of this model were presented in Ref. \cite{Maxim} in the form
 \be
\lb{Qsdual}
Q_\alpha = (\sigma_\mu \bar \psi)_\alpha ( p_\mu -   {\cal A}_\mu) \,, \ \ \ \ \ \ \ \ \ \ 
\bar Q^\alpha = (\psi \sigma_\mu^\dagger )^\alpha ( p_\mu -  {\cal A}_\mu)
 \ee
 with the following conventions 
(an Euclidean counterpart of the Wess and Bagger notations \cite{WB}):

{\it (i)} $(\sigma_\mu)_{\alpha \dot{\beta}} = \{i, \vecg{\sigma} \}_{\alpha \dot{\beta}}$,
$(\sigma_\mu^\dagger)^{ \dot{\beta} \alpha} = \{-i, \vecg{\sigma} \}^{\dot{\beta} \alpha }$. 

{\it (ii)} The fermion variables 
$\psi_{\dot{\alpha}}$ and $\bar\psi^{\dot{\alpha}} = (\psi_{\dot{\alpha}})^\dagger$ \  ($\{\bar\psi^{\dot{\alpha}},
\psi_{\dot{\beta}}\}  = \delta^{\dot{\alpha}}_{\dot{\beta}}$)    carry
only the dotted indices (in constrast to the supercharges \p{Qsdual} having undotted indices; in Euclidean space,
the $SU(2)$ groups acting on the dotted and undotted spinors are completely unrelated).

{\it (iii)} The indices are raised and lowered with $\epsilon_{\dot{\alpha} \dot{\beta}} = - 
\epsilon^{\dot{\alpha} \dot{\beta}}$.

  One can observe now that the supercharges \p{Qsdual} and hence the Hamiltonian 
commute with the $SU(2)$ generators,  
 \be
\hat L^a = 2t^a - i\eta^a_{\mu\nu} \left( x_\mu \partial_\nu + \frac 14 \psi \sigma^\dagger_\mu \sigma_\nu 
\bar \psi \right)   \, . 
  \ee
 Performing the Hamiltonian reduction with respect to  $\hat L^a$ gives us a 
({\bf 1}, {\bf 4}, {\bf 3}) model with only one dynamic bosonic degree of freedom. 
At the distances much larger than the instanton
size $\rho$ , the Hamiltonian thus obtained should go over to the conformal matrix 
Hamiltonian derived in \cite{IFL},
  \be
\lb{H143}
H = \frac 12 \left( p^2 + \frac 3{4x^2} \right) + \frac {2i t^a \, (\psi \sigma^a \bar \psi )}{x^2} \, .
 \ee
In that paper, also a 2-center model not enjoying the rotational symmetry was worked out. Probably, it can also be
obtained from a certain known model by applying two operations 
(similarity transformation and Hamiltonian reduction). It would
be interesting to see whether it is the case and, if yes, in a what particular way.
The same concerns many particular SQM models with ${\cal N} = 4$ and ${\cal N} = 2$ supersymmetries  
 constructed in recent \cite{BelKriv}. 
They were constructed using  ``semi-dynamic'' spin variables technique \cite{spinvar}.
The mathematical structure of these models, their {\it raison d'\^etre} is, however, not clear by now.  
It would be interesting to find out  by what particular operations from what particular 
known models are they  obtained.   

The OKT models with one or several root ({\bf 8}, {\bf 8}, {\bf 0}) -  multiplets also generate many different models 
after Hamiltonian reduction.
One of them is the beautiful   ${\cal N} = 8$  ({\bf 5}, {\bf 8}, {\bf 3}) - model with the metric 
$ds^2 = (1 + C/r^3) (dx^M)^2$ \cite{DE,sympl,ISsympl}
\footnote{It enjoys  $O(5) = Sp(4)$ symmetry and  has many common features with the   ({\bf 3}, {\bf 4}, {\bf 1}) model of 
Ref.\cite{1987} which is 3-dimensional and knows about $O(3) = Sp(2)$. That is why we called this type of sigma models
{\it symplectic}.}.  

There are many others.

\section{Gauge models}

It is known since Dirac that gauge theories can be interpreted as Hamiltonian systems 
involving first class constraints, the operators
$\hat G^a$ commuting with the Hamiltonian. One then performs  a Hamiltonian reduction with respect to $\hat G^a$ such that the
 Hilbert space of the large system 
is reduced to the small {\it physical} Hilbert space including only the wave functions 
annihilated under the action of $G^a$. This is rather
similar to the ideology of the present paper, but there is also an important dictinction.

 Consider a supersymmetric field theory  and assume that nothing depends on spatial coordinates 
(such dimensional reduction is, of course,
a variety of Hamiltonian reduction). We obtain a certain quantum mechanical system. To give a nontrivial enough 
but not too complicated example, consider the dimensionally reduced (2+1)-dimensional 
supersymmetric Yang-Mills model with $SU(2)$ gauge
group. The supercharges of the model are 
 \be
\lb{QSYM3}
 Q \ =\ \Pi_-^a \psi^a +  i B^a \bar \psi^a, \ \ \ \ \ \ \ \ \ \ \ \  \bar Q \ =\ \Pi_+^a \bar \psi^a -  i B^a  \psi^a \, ,
 \ee
where $\Pi_\pm^a = \Pi_1^a \pm i \Pi_2^a$  are holomorphic combinations of canonical momenta, 
$\psi^a$ and $\bar \psi^a$ are canonically conjugated
fermion variables, and
  \be
\lb{B}
B^a \ =\ \epsilon^{abc} \epsilon_{jk} A^b_j A^c_k = - \frac i2 \epsilon^{abc} A_-^b A_+^c 
 \ee
is the non-Abelian magnetic field strength. The coupling constant dependence is suppressed by choosing proper units.
  
The Hamiltonian $H = \{\bar Q, Q \}$ has the form 
  \be
\lb{HamSYM3}
H \ =\ \frac {1}2 (\Pi_j^a)^2 + \frac 14 [(A^a_j A^a_j)^2 - A^a_j A^a_k A^b_j A^b_k]    
+ \frac {i \epsilon^{abc}}2 [\bar \psi^a \bar \psi^b A^c_+ +  \psi^a  \psi^b A^c_-]
 \ee
 Note, however, that the supercharges written in \p{QSYM3} are not nilpotent. One easily derives
$Q^2 = A^a_- \hat G^a$, where
 \be
\lb{Ga}
 \hat G^a = \epsilon^{abc} (A_j^b \Pi_j^c - i \psi^b \bar \psi^c)
 \ee
are Gauss law constraints - generators of gauge transformations. If we want to keep supersymmetry, one {\it should}
perform the Hamiltonian reduction and impose the constraint $ \hat G^a \Psi = 0$.

One can now {\it resolve} the constraints, i.e. get rid of three variables on which nothing depends (gauge degrees of freedom) 
and to write
the Hamiltonian in  reduced phase space. For field theories, this is practically impossible, but, 
for quantum mechanical systems, it is
quite feasible. It is convenient to use the polar representation for the vector potential \cite{polar}. 
In the (2+1)-dimensional case, 
it boils down to
\be
\label{polar}
A^{\ a}_j \ =\ U_{jk} \Lambda^{\ b}_k V_{ba} \, ,
 \ee
  where $U_{jk} (\alpha)$ is an $O(2)$ matrix describing spatial rotations, $V_{ba} (\phi^a)$ is an $O(3)$ gauge
rotation matrix and $\Lambda^{\ b}_k$ is a quasidiagonal matrix,
\be
\label{Lambda}
\Lambda_k^{\ b} \ =\ \left( \begin{array}{ccc} a & 0 & 0 \\ 0 & b & 0 \end{array} \right)\ .
 \ee  
Choosing the gauge $V_{ba} = \delta_{ba}$, we are left with just three gauge invariant bosonic variables: $a,b,\alpha$.
In addition,  reduced phase space inherits  all three  complex fermion variables $\psi^a, \bar \psi^a$. 
The explicit expressions for the reduced
supercharges and Hamiltonian are rather complicated \cite{oncemore}. We will present here only the supercharges.
  \be
\label{Qresolved}
Q^{\rm cov} = e^{-i\alpha} g_0 \left[  \psi^1 \left( p_a - 
 i\frac {a p_\alpha +  b J^3 }{a^2 - b^2} \right) +
\psi^2 \left( -ip_b  + \frac {b p_\alpha +   
a J^3 }{a^2 - b^2} \right) \right. \nonumber \\
\left. - \psi^3 \left( \frac {J^2}a + \frac {i J^1}b \right)  \right] + 
\frac {iab}{g_0} \bar \psi^3 \ , \nonumber \\
\bar Q^{\rm cov} =  g_0 e^{i\alpha} \left[   \bar \psi^1 \left( 
p_a  + i\frac {a p_\alpha + 
 b J^3}{a^2 - b^2}  \right) +
 \bar \psi^2 \left( ip_b + \frac {b p_\alpha   + a J^3}{a^2 - b^2}  \right)  \right. \nonumber \\
  -  \left.  \bar \psi^3  \left( \frac {J^2 }a - \frac {iJ^1}b \right)\right]
 - \frac {iab}{g_0} \psi^3 \ ,
 \ee
where $J^a = i\epsilon^{abc} \psi^b \bar \psi^c$.

This is a kind of sigma model, the metric in the space $\{a,b,\alpha\}$ induced by the flat metric in the space $\{A^a_j\}$ 
being nontrivial.  
The configuration space involves a complex fermion variable for each bosonic variable, i.e. the number of degrees of freedom 
is the same as for the de Rham complex. But it is not a de Rham system: for the latter the fermion charge is conserved, while
the supercharge $Q$ in \p{Qresolved} involves the terms $\propto \bar \psi$ on top of the terms $\propto \psi$.
   
  If our conjecture is true, the supercharges \p{Qresolved} can be obtained by a similarity 
transformation and Hamiltonian reduction 
from a free system. However, in this case, a pure similarity transformation of  the system \p{Qreal} 
of real dimension 3 would probably be 
not sufficient. Indeed, the only imaginable to us way to obtain the term 
$\propto \bar \psi^3$ out of  a "flat" supercharge
like $p_a \psi^{(a)} + p_b \psi^{(b)} + p_\alpha \psi^{(\alpha)}$  is applying 
an antiholomorphic transformation \p{antihol}. 
But such a transformation can generate only  the terms $\sim \bar\psi$ that multiply
canonical momenta and in addition the unwanted terms $\propto \psi \bar \psi \bar \psi$.

Thus, to derive \p{Qresolved}, one should start from a free system of  larger dimension. Indeed, as was discussed
above, the system \p{Qresolved} can be obtained by a Hamiltonian reduction of a more simply looking system
\p{QSYM3}, \p{HamSYM3} with extended phase space. But, in constrast to the examples with Hamiltonian reduction 
discussed in the previous section, such extended system
{\it is} not supersymmetric --- this is the distinction that we were talking about. The absence of supersymmetry 
in the large system can be traced back to the fact that, when writing \p{QSYM3}, \p{HamSYM3}, we have already partially 
fixed the gauge   (Wess-Zumino gauge) and got rid of some number of components in the spinor superfield 
$\Gamma_\alpha$ describing $3D$ SYM theory. 
Supersymmetry is {\it broken} by such partial gauge fixing and is restored when the gauge is fixed {\it completely}.
The question is thus reduced to the question whether this large SQM system involving all components of $\Gamma_\alpha$  
can be obtained by our recipe. We hope to address it in later studies.

The same set of question can be asked to a system obtained by the dimensional reduction of (3+1) - dimensional, 
(5+1)-dimensional, 
or (9+1)-dimensional SYM theory with or without extra matter multiplets.  For example,  pure (3+1) SYM theory involves
$3\cdot 3 - 3 = 6$ gauge invariant variables. The explicit expressions for its reduced supercharges and  
Hamiltonian are given in Ref.\cite{howto}. 

\section{Field theories.}

A natural question to ask is whether our conjecture formulated for SQM systems works also for field theories.

The fast answer that it should, because  quantum field theory is nothing but a SQM system with 
an infinite number of degrees of freedom.  
The devil is, as usual, in the details. 

Consider the simplest example - the 4D Wess-Zumino model. To make the things still simpler, 
let it be {\it free massless} WZ model with the Lagrangian
 \be
\lb{LWZ}
 L \ =\ \int d\vecg{x} \, \left[\partial_\mu \bar \phi \partial_\mu  \phi + 
i \psi \sigma_\mu \partial_\mu \bar \psi \right] \,  
 \ee
where $\psi^\alpha,\, \bar\psi^{\dot{\alpha}} $ are complex conjugate {\it Minkowskian} 
Weyl spinors and $(\sigma_\mu)_{\alpha \dot{\beta}} = 
(1, \sigma_j)_{\alpha \dot{\beta}}$. 
\footnote{Our conventions are almost the same as in Ref. \cite{WB}, but the metric is chosen with the opposite sign,
$\eta_{\mu\nu} = {\rm diag} (1,-1,-1,-1)$.} 
The corresponding Hamiltonian
 \be
\lb{HWZ}
 H \ =\ \int d\vecb{x} \, \left[ \bar \Pi  \Pi +  \partial_j \bar \phi \partial_j \phi  
- i \psi \sigma_j \partial_j \bar \psi \right] \,  
 \ee
is supersymmetric. There is the Weyl doublet of supercharges,
  \be
\lb{QWZ}
 Q_\alpha  \ =\ \sqrt{2} \int d\vecb{x} \, \left[\Pi \psi_\alpha +  \partial_j \bar \phi \,
 (\sigma_j)_{\alpha \dot{ \gamma}} \delta^{\dot{\gamma} \gamma} \psi_\gamma\right]\, , \nn
 \bar Q_{\dot{\alpha}}  \ =\ \sqrt{2} \int d\vecb{x} \, \left[\bar \Pi \bar \psi_{\dot{\alpha}}  +  
(\partial_j  \phi ) \bar \psi_{\dot{\gamma}}
\delta^{\dot{\gamma} \gamma}
(\sigma_j)_{\gamma \dot{\alpha}}  \right]\,  
 \ee
 They satisfy the algebra
 \be
\lb{algQWZ}
 \{Q_\alpha, \bar Q_{\dot{\alpha}}\}  \ =  \ 2(\sigma_\mu)_{\alpha \dot{\alpha}} P_\mu  = \ 
2\left[ \delta_{\alpha \dot{\alpha}}H 
 + (\sigma_j)_{\alpha \dot{\alpha}} P_j\right] \, ,
 \ee
where $\vecg{P}$ is the 3-momentum operator.

Put the system in a finite box of size $L$, which we set to 1, and expand 
$\phi(\vecb{x}), \psi(\vecb{x})$  in the Fourier series,
 \be
 \phi(\vecb{x}) \ =\ \sum_{\vecb{n}} \phi_{\vecb{n}} e^{2\pi i \vecb{n} \vecb{x}}, \ \ \ \ 
\psi(\vecb{x}) \ =\ \sum_{\vecb{n}} \psi_{\vecb{n}} e^{2\pi i \vecb{n} \vecb{x}}\, \nn
\bar \phi(\vecb{x}) \ =\ \sum_{\vecb{n}} \bar \phi_{\vecb{n}} e^{-2\pi i \vecb{n} \vecb{x}}, \ \ \ \ 
\bar \psi(\vecb{x}) \ =\ \sum_{\vecb{n}} \bar \psi_{\vecb{n}} e^{-2\pi i \vecb{n} \vecb{x}}\,
  \ee
 The Hamiltonian \p{HWZ} is expressed via the modes 
as follows,
\be
\lb{HWZmodes}
 H \ =\ \sum_{\vecb{n}} \, \left[ \bar \Pi_{\vecb{n}}   \Pi_{\vecb{n}}  + \left( 2\pi \vecb{n} \right)^2
 \bar \phi_{\vecb{n}}   \phi_{\vecb{n}}  -  
   2\pi n_j  \psi_{\vecb{n}}   \sigma_j  \bar \psi_{\vecb{n}}   \right] 
 \ee
($\Pi_{\vecb{n}} = \dot{\bar \phi_{\vecb{n}}}$) 
and the supercharges are
 \be
\lb{QWZmodes}
Q_\alpha &=& \  \sqrt{2} \sum_{\vecb{n}}  \left[ \Pi_{\vecb{n}} \psi_{\alpha \vecb{n}} - 2\pi i n_j \, 
( \sigma_j)_{\alpha \beta}
\psi_{\beta \vecb{n}}
 \bar  \phi_{\vecb{n}}  \right] \, , \nn
\bar Q_\alpha  &=& \  \sqrt{2} \sum_{\vecb{n}}  \left[ \bar \Pi_{\vecb{n}} \bar \psi_{\alpha \vecb{n}} +
 2\pi i n_j \,    \bar \psi_{\beta \vecb{n}} (\sigma_j)_{\beta \alpha }
   \phi_{\vecb{n}}   \right]
 \ee
(The finite box breaks Lorentz invariance and there is no point to distinguish 
the usual and dotted indices anymore).

This is an SQM model with an infinite number of degrees of freedom, indeed. 
One can observe, however, that  from the SQM viewpoint,
\begin{itemize}
\item  it is not a {\it basic} system  \p{TTbar} as, besides the terms $\propto \Pi \psi$, the supercharges
involve extra terms.
\item the supercharges \p{QWZmodes} satisfy not the standard ${\cal N} = 4$ superalgebra, 
but the algebra \p{algQWZ} 
involving the 3-momentum playing the role of a central charge.
\end{itemize}

Furthermore, the supercharges \p{QWZmodes} do {\it not} seem to be related to the basic supercharges 
$  \sum_{\vecb{n}}   \Pi_{\vecb{n}} \psi_{\alpha \vecb{n}}, \ 
 \sum_{\vecb{n}}  \bar \Pi_{\vecb{n}} \bar \psi_{\alpha \vecb{n}} $ by a similarity transformation. 

It is still possible to write down a similarity transformation of the 
{\it de Rham} free (in the SQM sense) supercharge 
${\cal Q} = p_A\psi_A$ such that 
the anticommutator of the transformed supercharge ${\cal Q}$ and its conjugate would give \p{HWZmodes}. 

Consider one particular term $H_{\vecb{n}}$ in the sum \p{HWZmodes}. 
Let first $\vecb{n} \neq 0$. One can observe that the matrix $n_j \sigma_j$ has two 
eigenvalues $\lambda_{1,2} = \pm \sqrt{\vecb{n}^2}$.  
If denoting by $\chi^{1,2}_{\vecb{n}}$ the corresponding normalized eigenvectors, one can represent
 \be
\lb{eigen}
n_j  \psi_{\vecb{n}}   \sigma_j  \bar \psi_{\vecb{n}}  \ 
= \ \sqrt{\vecb{n}^2}  (\chi^1_{\vecb{n}} \bar \chi^1_{\vecb{n}} -
\chi^2_{\vecb{n}} \bar \chi^2_{\vecb{n}})
  \ee
Then one can define
 \be
\lb{calQmodes}
 {\cal Q}_{\vecb{n}} = 
 \chi^1_{\vecb{n}} \left( P^1_{\vecb{n}} + 2i\pi f^1_{\vecb{n}} \sqrt{\vecb{n}^2}  \right)
+  \chi^2_{\vecb{n}} \left( P^2_{\vecb{n}} - 2i\pi f^2_{\vecb{n}} \sqrt{\vecb{n}^2}  \right) 
 \ee
($P^{1,2}_{\vecb{n}}/\sqrt{2}$ and $f^{1,2}_{\vecb{n}}/\sqrt{2}$ being the real and imaginary parts of
$\Pi_{\vecb{n}}$ and $\phi_{\vecb{n}}$). The operator \p{calQmodes}  is nilpotent and 
 \be
\lb{Hn}
 \{{\cal Q}_{\vecb{n}}, \bar {\cal Q}_{\vecb{n}} \}   \ =\ 2 H_{\vecb{n}} \, .
 \ee
This holds also in the case of degenerate eigenvalues, $\vecb{n} = 0$, if choosing for $\chi^{1,2}_{\vecb{n}}$ arbitrary
orthonormal vectors. 

Obviously, the full supercharge 
 \be
\lb{QsumQn}
{\cal Q} = \sum_{\vecb{n}} {\cal Q}_{\vecb{n}}
 \ee
is also nilpotent, and 
  $\{ {\cal Q}, \bar {\cal Q} \}/2$ gives the Hamiltonian \p{HWZmodes}.
In fact, this model represents a multidimensional (with an infinity of degrees of freedom)
generalization of the model \p{QWit}, \p{HWit} with the superpotential
 \be
\lb{Wsumn}
W =  \sum_{\vecb{n}} \pi \sqrt{\vecb{n}^2} \left[ (f^1_{\vecb{n}})^2 - (f^2_{\vecb{n}})^2 \right]
 \ee
Hence, the
 supercharge \p{QsumQn} can be  related to the "free" supercharge
\be
\lb{Q0}
{\cal Q}^{(0)} = \sum_{\vecb{n}} \left( P^1_{\vecb{n}} \chi^1_{ \vecb{n}} +  P^2_{\vecb{n}} \chi^2_{ \vecb{n}} \right) 
 \ee
by the similarity transformation \p{simil}.
 
Thus, we showed that, when expanded over the modes, the field theory  \p{LWZ} can be obtained by a 
similarity transformation from the "free" supercharge \p{Q0}, as it should according to our conjecture.

The  problem, however, is
 that this transformation and both the supercharges 
\p{Q0} and  \p{calQmodes} are highly nonlocal. None of them does  have therefore a lot of physical sense.

It might be more reasonable to treat the Lorentz-invariant model \p{LWZ} as the free one and ask whether the supercharges 
of the {\it interacting} WZ model could be obtained from the free supercharges \p{QWZ} by a similarity transformation.
Unfortunately, the answer to this question seems to be negative.  

Indeed, the interacting WZ supercharges are obtained
from the free supercharges \p{QWZ} by adding the following extra terms, 
   \be
\lb{QWZint}
 Q_\alpha^{\rm int}   \ =\  Q_\alpha^{\rm free} + i \sqrt{2} \int d\vecb{x} \, {\cal W}'(\bar \phi) 
\delta_{\alpha \dot{\alpha}} \bar \psi^{\dot{\alpha}}
 \, , \nn
   \bar Q_{\dot{\alpha}}^{\rm int}   \ =\  \bar Q_{\dot{\alpha}}^{\rm free} - i \sqrt{2} \int d\vecb{x} \, 
{\cal W}'( \phi) 
\delta_{\dot{\alpha} \alpha }  \psi^{\alpha }
 \, ,
 \ee
 where ${\cal W}(\phi)$ is the WZ superpotential (having nothing to do
with \p{Wsumn}). And we do not see how to obtain $Q_\alpha^{\rm int}$ out of $Q_\alpha^{\rm free}$ 
by a similarity transformation. The problem is the same as with the supercharge $Q^{\rm cov}$ in Eq. \p{Qresolved}. 
The only known to us way to generate the term $\propto \bar\psi$ in the supercharge $Q$ is to apply an antiholomorphic 
transformation, like in  \p{antihol}. But such a transformation would produce the terms where $\bar\psi$ 
is multiplied by $\Pi$ or else the terms $\propto \bar\psi\bar\psi \psi$...

\section{Discussion and outlook.}
Our main point is the 

\vspace{1mm}

{ \bf Conjecture}. {\it Any SQM model can be related to a free complex model \p{freed} by a combination of two operations:
(i) similarity transformation of properly chosen complex supercharges and (ii) Hamiltonian reduction.}

\vspace{1mm}

We have not proven it, but checked in many nontrivial examples. In particular, we discussed nontrivial 
sigma models with extended
supersymmetries and showed that, for
the K\"ahler de Rham sigma models, hyper-K\"ahler de Rham sigma models, and HKT models, 
all complex supercharges are derived from the
free supercharges by the {\it same} similarity transformation.

On the other hand, we have {\it not} seen yet that this conjecture also works for gauge SQM models. We noted that this question
 can be clarified if analyzing the supercharges and the Hamiltonian of gauge models 
{\it before} gauge is fixed such that supersymmetry
is realized linearly.

In Sect. 5, we discussed field theories and found out that, though our recipe seems to work (it works in the simplest 
case that we analyzed), the similarity transformation turns out to be highly nonlocal and therefore useless.

The last remark is the following. Philosophically,  similarity transformations considered in this paper
remind the Nicolai map \cite{Nicolai}. In both cases, an interacting model is related to a free one. 
However, the ways they are  
related are rather different. The Nicolai map is a nonlocal transformation of {\it bosonic} 
variables that renders the functional 
integral for the index Gaussian allowing one to do it. For the simplest nontrivial SQM model \p{HWit}, 
it amounts to the change
 \be
\lb{Nicolai}
 \dot{x} \pm W'(x) \ \to  \ \dot{y}
  \ee
It is not similar to the local transformation of the supercharges studied in this paper, 
though more meditations in this direction are definitely welcome.  

I am indebted to S. Fedoruk and E. Ivanov for many illuminating discussions.

\section*{Appendix: K\"ahler and hyper-K\"ahler superalgebras.}

We will prove here some well-known to mathematicians facts \cite{Verb} in the SQM language understandable to physicists.

We will be interested in extra supersymmetries of the de Rham ($\{{\bf 1}, {\bf 2}, {\bf 1}\}$) 
sigma models that come into existence when the manifold is K\"ahler.
We are using the method that was used earlier to study extra supersymmetries for the Dolbeault 
($\{{\bf 2}, {\bf 2}, {\bf 0}\}$)) models 
for hyper-K\"ahler \cite{Wipf} and HKT \cite{HKT} manifolds.

We start with reminding

\vspace{.2cm}

{\bf One of the possible definitions.}  The complex
 manifold is called K\"ahler if its complex structure tensor,
\footnote{For the manifold to be genuinely {\it complex} and not just {\it almost complex}, the tensor
$I_{MN}$ should satisfy besides \p{complstruc} also a certain integrability condition. But if the tensor $I$ is covariantly
constant, this condition is satisfied automatically.} 
 \be
\lb{complstruc}
I_{MN} = -I_{NM}, \ \ \ \ \ \ \ \ \ \ I_M^{\ P} I_P^{\ N} = - \delta_M^{\ N}
 \ee
 is covariantly constant,
   \be
\lb{nablaI}
{\cal D}_P I_{MN} = \partial_P I_{MN} - \Gamma^S_{PM} I_{SN} - \Gamma^S_{PN} I_{MS} =  0\ .
 \ee 

{\bf Similarly}:   The  manifold is called hyper-K\"ahler if it admits three different covariantly constant 
complex structures $I^a$ satisfying the quaternionic algebra \p{quatern}.  

\vspace{1mm}

We will prove now two theorems.

\begin{thm}
 If the manifold is K\"ahler, the supercharges \p{QRham}, \p{Srot} and their conjugates
satisfy the 
${\cal N} =4$ superalgebra with the only nonvanishing anticommutators
 \be
\lb{N4alg}
\{Q, \bar Q\} = \{S, \bar S\}  \ .
 \ee
\end{thm}

\begin{proof}
{\it (i)} Nilpotency of $Q, S$ and the property $\{Q, S\} = 0$ follow from the proven above fact that 
$Q$ and $S$ are obtained by the same
similarity transformation
of the correponding flat supercharges and the validity of the ${\cal N} =4$ superalgebra for the latter.

{\it (ii)} To deal with the commutators like $\{Q, \bar S\}$, introduce the operators 
 \be
\lb{F+-}
 F_+ = \frac 12 I_{MN} \bar \psi^M  \bar \psi^N, \ \ \ \ F_- = \frac 12 I_{MN} \psi^M  \psi^N \ .
 \ee
Their commutator gives the  fermion charge operator, $F_0 = \psi_M \bar \psi^M$. The operators $\{F_-, F_0, F_+\}$ form
an $SU(2)$ triplet. 

{\it (iii)}
Consider now the commutator $[Q, F_+]$. Capitalizing on the scalar nature of $F_+$, one can upgrade the ordinary derivatives
in the combination $\partial_M  + \Omega_{M, AB} \psi_A \bar \psi_B $ to the covariant ones, $\partial_M \to 
{\cal D}_M$. The supercharge \p{QRham} acquires then the form $-i\psi^M \nabla_M$, where $\nabla_M$ is the
{\it full} covariant derivative involving also the fermion (spinor) part.

  Then one notes that
$$ \nabla_M \bar \psi^N \ =\ ({\cal D}_M e^N_A) \bar \psi_A - \Omega_{M, BA}e^N_B \, \bar \psi_A \ =\ 0$$
and uses the  condition \p{nablaI} that the manifold is K\"ahler to derive
 \be
 [Q, F_+] = \ 
i \bar \psi^Q I_Q^{\ M} \nabla_M = - \bar S \ , \ \ \ \ \ \  [\bar Q, F_-] \ 
 =\   -S \, , \nonumber
 \ee
\be
[S, F_+]  = \bar Q , \ \ \ \ \ \ \ \ \ \ \ \ \ \ \ \ \ [\bar S, F_-] =  Q 
 \ee
(the commutators $[Q, F_-]$ and $[\bar Q, F_+]$ vanish).
The vanishing of $\{Q, \bar S\} = \{Q, [F_+, Q]\}$ follows from nilpotency of $Q$ and the Jacobi identity.  

{\it (iv)} The anticommutator $\{S, \bar S\} = \{S, [F_+, Q]\}$ is reduced to the 
anticommutator $\{[S, F_+], Q\} = \{\bar Q, Q\}$ by the
Jacobi identity.
\end{proof}

\begin{thm}
 If the manifold is hyper-K\"ahler, the supercharges \p{QRham}, \p{Sarot} and their conjugates
satisfy the 
${\cal N} =8$ superalgebra with the only nonvanishing anticommutators
 \be
\lb{N8alg}
\{Q, \bar Q\} = \{S^1, \bar S^1\} =\  \{S^2, \bar S^2\} = \  
\{S^3, \bar S^3\}  \, .
 \ee
\end{thm}

\begin{proof}
{\it (i)} Bearing in mind the results of the previous theorem, we have only to prove that 
$\{S^a, S^b\} = 0$ and $\{S^a, \bar S^b\} = 0$ when $a \neq b$. 
The first equality follows from the fact that all $S^a$ are obtained from the flat 
holomorphic supercharges in \p{S123} by one and the same similarity transformation.

{\it (ii)} To calculate  $\{S^a, \bar S^b\}$, introduce the operators
 \be
\lb{Fa+-}
F^a_+ \ =\  \frac 12 I^a_{MN} \bar \psi^M  \bar \psi^N \, , \ \ \ \ \ \ \ \ F^a_- \ =\  \frac 12 I^a_{MN} \psi^M   \psi^N \ .
  \ee
 The same reasoning as above and the quaternionic algebra \p{quatern} allow one to derive
      \be
\lb{FQabc}
[S^a, F^b_+] = \delta^{ab} \bar Q - \epsilon^{abc} \bar S^c, \ \ \ \ \ [\bar S^a, F^b_-] = \delta^{ab} Q - \epsilon^{abc} S^c 
 \ee
and $[S^a, F^b_-] = [\bar S^a, F^b_+] = 0$. Then e.g. $\{S^1, \bar S^2\} = \{S^1, [S^1, F_+^3]\} $, 
which vanishes due to nilpotency of $S_1$ and the Jacobi identity.
\end{proof}

\end{document}